%% file: main.tex
\documentclass[runningheads]{llncs}
\usepackage{graphicx}
\usepackage{tikz} 
  \usepackage{tikz-qtree}
\usetikzlibrary{calc,snakes,shapes,arrows.meta}
\usepackage{amssymb,amsmath}  
\usepackage{subcaption}
\usepackage{hyperref}
\usepackage[capitalise]{cleveref}
\usepackage{verbatim}
\usepackage{enumerate}
\usepackage{todonotes}
\usepackage{multirow}
\usepackage[ruled,noend]{algorithm2e}
\DontPrintSemicolon

\usepackage[noend]{algpseudocode}
\newcommand{\cO}{\mathcal{O}}
\newcommand{\Oh}{\mathcal{O}}

\def\dd{\mathinner{.\,.}}

\newcommand{\Covers}{\mathsf{Covers}}
\newcommand{\SCovers}{\mathsf{SCovers}}
\newcommand{\LCovers}{\mathsf{LCovers}}

\def\pillar{{\tt PILLAR}\xspace}
\newcommand{\LCP}{\mathsf{LCP}}
\newcommand{\LCPR}{\mathsf{LCP}_{\mathsf{R}}}
\newcommand{\IPM}{\mathsf{IPM}}
\newcommand{\Cov}{\mathsf{Cov}}
\newcommand{\Fib}{\mathsf{Fib}}

\begin{document}
\title{Computing String Covers in Sublinear Time}
%
%

\author{
Jakub Radoszewski\inst{1}\fnmsep\thanks{Supported by the Polish National Science Center, grant no. 2022/46/E/ST6/00463.}\orcidID{0000-0002-0067-6401} \and
Wiktor Zuba\inst{2}\fnmsep\thanks{Supported by the European Union's Horizon 2020 research and innovation programme under the Marie Skłodowska-Curie Grant Agreement No 101034253.}\orcidID{0000-0002-1988-3507}
}
\authorrunning{J.~Radoszewski and W.~Zuba}
%
\institute{
University of Warsaw, Warsaw, Poland
\email{jrad@mimuw.edu.pl}
\and
CWI, Amsterdam, The Netherlands
\email{wiktor.zuba@cwi.nl}
}
\maketitle              
\begin{abstract}
Let $T$ be a string of length $n$ over an integer alphabet of size $\sigma$.
In the word RAM model, $T$ can be represented in $\Oh(n /\log_\sigma n)$ space.
We show that a representation of all covers of $T$ can be computed in the optimal $\Oh(n/\log_\sigma n)$ time; in particular, the shortest cover can be computed within this time.
We also design an $\Oh(n(\log\sigma + \log \log n)/\log n)$-sized data structure that computes in $\Oh(1)$ time any element of the so-called (shortest) cover array of $T$, that is, the length of the shortest cover of any given prefix of $T$.
As a by-product, we describe the structure of cover arrays of Fibonacci strings.
On the negative side, we show that the shortest cover of a length-$n$ string cannot be computed using $o(n/\log n)$ operations in the \pillar model of Charalampopoulos, Kociumaka, and Wellnitz (FOCS 2020). 

\end{abstract}

\section{Introduction}
A string $C$ is called a \emph{cover} (or a \emph{quasiperiod}) of a string $T$ if each position in $T$ lies within an occurrence of $C$ in $T$. A cover is called \emph{proper} if it is shorter than the covered string. A string that does not have proper covers is called \emph{superprimitive} (see \cite{DBLP:journals/ipl/Breslauer92}). The shortest cover of a string of length $n$ can be computed in $\Oh(n)$ time~\cite{DBLP:journals/ipl/ApostolicoFI91}. Furthermore, all covers of a length-$n$ string can be computed in $\Oh(n)$ time~\cite{DBLP:journals/ipl/MooreS95}. A cover of a string is a prefix of the string, so a string of length $n$ indeed has at most $n$ covers.

The lengths of all covers of a string of length $n$ can be represented using $\Oh(\log n)$ disjoint arithmetic progressions
~\cite{DBLP:conf/spire/CrochemoreIRRSW20}. For a string $T$, we denote such a representation as $\Covers(T)$. A similar representation is well known to exist for the set of all borders of a string (see, e.g., \cite{DBLP:conf/cpm/CrochemoreIKKRRTW12}).

We consider the standard word RAM model with machine word composed of $\omega \ge \log_2 n$ bits.
In this model, a string of length $n$ over an alphabet of size $\sigma$ can be represented using $\Oh(n/\log_\sigma n)$ machine words, that is, $\Oh(n \log \sigma)$ bits, in a so-called packed representation; see~\cite{DBLP:journals/tcs/Ben-KikiBBGGW14}.
In \cref{sec:sublinear} we show the following result that improves upon~\cite{DBLP:journals/ipl/ApostolicoFI91,DBLP:journals/ipl/Breslauer92,DBLP:journals/ipl/MooreS95} in the case that the string is over a small alphabet.

\begin{theorem}\label{thm:main1}
A representation $\Covers(T)$ of all the covers of a string $T$ of length $n$ over an alphabet of size $\sigma$ given in a packed form, consisting of $\Oh(\log n)$ arithmetic progressions, can be computed in $\Oh(n /\log_\sigma n)$ time.
\end{theorem}

The representation $\Covers(T)$ can be transformed in $\Oh(n/\log n)$ time to a Boolean array of size $n$, represented in a packed form, that stores for every $\ell \in [1 \dd n]$ a Boolean value that determines if a length-$\ell$ prefix of $T$ is a cover of $T$.

\begin{figure}
\centering
\begin{tikzpicture}[scale=0.5]
\draw (0,0) node[left] {$\ell$:};
\draw (0,1.8) node[left] {$T[\ell-1]$:};
\draw (0,0.8) node[left] {$\Cov_T[\ell]$:};
\foreach \x/\y in {0/0,3/0,5/1,8/0,11/0,13/1,16/0,18/1}{
\filldraw[white!90!black,xshift=\x cm,yshift=\y cm] (0.6,2.4) rectangle (3.4,3.3);
\draw[xshift=\x cm,yshift=\y cm] (0.6,2.4) rectangle (3.4,3.3);
\foreach[count=\i] \c in {a,b,a}{
\draw[xshift=\x cm,yshift=\y cm] (\i,2.4) node[above] {\texttt{\c}};
}
}
\foreach \x/\y in {0/0,8/0,13/1}{
\filldraw[white!95!black,xshift=\x cm,yshift=\y cm] (0.6,4.9) rectangle (8.4,5.8);
\draw[xshift=\x cm,yshift=\y cm] (0.6,4.9) rectangle (8.4,5.8);
\foreach[count=\i] \c in {a,b,a,a,b,a,b,a}{
\draw[xshift=\x cm,yshift=\y cm] (\i,4.9) node[above] {\texttt{\c}};
}
}
\foreach[count=\i] \c in {\textcolor{white!60!black}{1},\textcolor{white!60!black}{2},\textcolor{white!60!black}{3},\textcolor{white!60!black}{4},\textcolor{white!60!black}{5},3,\textcolor{white!60!black}{7},3,\textcolor{white!60!black}{9},5,3,\textcolor{white!60!black}{12},5,3,\textcolor{white!60!black}{15},3,9,5,3,\textcolor{white!60!black}{20},3}{\draw (\i,0.5) node[above] {\c};}
\foreach[count=\i] \c in {a,b,a,a,b,a,b,a,a,b,a,a,b,a,b,a,a,b,a,b,a}{\draw (\i,1.5) node[above] {\texttt{\c}};}
\foreach \i in {1,2,...,21}{\draw (\i,0) node {\scriptsize \i};}
\end{tikzpicture}
\caption{Both proper covers (\texttt{aba}, \texttt{abaababa}) and the cover array of a Fibonacci string $T$. Values $\Cov_T[\ell]=\ell$ corresponding to superprimitive prefixes $T[0 \dd \ell)$ are shown in gray.}
\label{fig:cov_fib}
\end{figure}
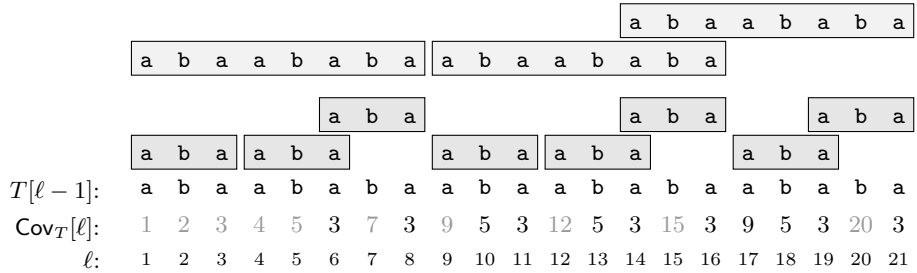

The (shortest) cover array of a string $T$, $\Cov_T[1 \dd |T|]$, stores for every position $\ell$ of $T$ the length of the shortest cover of a length-$\ell$ prefix of $T$ as $\Cov_T[\ell]$; see \cref{fig:cov_fib}. The cover array is the output of Breslauer's on-line algorithm computing shortest covers~\cite{DBLP:journals/ipl/Breslauer92}; see also~\cite{DBLP:conf/cpm/CrochemoreIPT10}. We give a sublinear-sized representation of this array.

\begin{theorem}\label{thm:online}
Let $T$ be a text of length $n$ over an integer alphabet of size $\sigma$. There exists a data structure using space $\Oh(n(\log\sigma + \log \log n)/\log n)$ that, given $\ell \in [1 \dd n]$, returns $\Cov_T[\ell]$ in $\Oh(1)$ time.
\end{theorem}

Our results extend the list of basic stringology problems for which representing the input in a packed form allows to obtain an $o(n)$-time solution; see~\cite{DBLP:conf/esa/BannaiE23,DBLP:journals/tcs/Ben-KikiBBGGW14,DBLP:conf/esa/Charalampopoulos21,DBLP:conf/cpm/Charalampopoulos22,DBLP:conf/stoc/KempaK19,DBLP:conf/cpm/MunroNN20}.

As a by-product, we give a characterization of the cover arrays of Fibonacci strings (\cref{thm:fib}).

We also consider covers in the \pillar model. This model was introduced in~\cite{DBLP:conf/focs/Charalampopoulos20} with the aim of unifying approximate pattern matching algorithms across different settings.
In this model, we consider a collection $\mathcal{X}$ of strings and assume that certain primitive \pillar operations
can be performed efficiently. The set of primitive operations consists of computing the length of the longest common prefix ($\LCP$) or suffix ($\LCPR$) of substrings of strings in $\mathcal{X}$, so-called internal pattern matching ($\IPM$) queries that ask for the set of occurrences of one substring in another substring that is at most twice as long, represented as an arithmetic progression, as well as simple operations allowing to access letters of strings. (For a formal definition, see \cref{sec:pillar}.)

The strength of the \pillar model lies in the fact that efficient implementations of its primitives are known in many different settings:
\begin{itemize}
    \item In the \emph{standard setting}, in which all strings in the collection $\mathcal{X}$ are substrings of a given string of length $n$ over an integer alphabet of size $\sigma$, each \pillar operation on its substrings can be performed in $\Oh(1)$ time after $\Oh(n)$ preprocessing~\cite{DBLP:conf/latin/BenderF00,DBLP:conf/focs/Farach97,DBLP:conf/soda/KociumakaRRW15} and even after just $\Oh(n/\log_\sigma n)$ preprocessing~\cite{DBLP:conf/stoc/KempaK19,DBLP:journals/corr/KociumakaRRW13}.
    \item In the \emph{dynamic setting}, the collection $\mathcal{X}$ can updated dynamically under edit operations (insertions, deletions, substitutions) with each edit operation and each \pillar operation performed in $\Oh(\log^{\Oh(1)}N)$ time, where $N$ is the total size of $\mathcal{X}$~\cite{DBLP:conf/focs/Charalampopoulos20,DBLP:conf/stoc/KempaK22}.
    \item In the \emph{fully compressed setting}, given a collection $\mathcal{X}$ of straight-line programs (SLPs) of total size $n$ generating strings of total length $N$, each \pillar operation can be performed in $\Oh(\log^2 N \log \log N)$ time after $\Oh(n \log N)$-time preprocessing~\cite{DBLP:conf/focs/Charalampopoulos20}.
    \item An efficient implementation of the \pillar operations is also known in the \emph{quantum setting}~\cite{DBLP:journals/jda/HariharanV03,DBLP:conf/soda/JinN23}.
\end{itemize}

Thus if a problem can be solved fast in the \pillar model, it immediately implies its efficient solutions in all the above mentioned settings. For example, the fact that an $\Oh(\log n)$-sized representation of all the periods (equivalently, borders) of a length-$n$ string can be computed in $\Oh(\log n)$ time in the \pillar model~\cite[Period Query]{DBLP:journals/corr/KociumakaRRW13,DBLP:conf/soda/KociumakaRRW15} implies that a representation of the periods of a dynamic string can be updated in $\Oh(\log^{\Oh(1)} N)$ time per operation and that a representation of all periods of a fully compressed string of length $N$ generated by an SLP of size $n$ can be computed in $\Oh(n \log^{\Oh(1)} N)$ time. In the case of covers, some efficient algorithms were designed for each of the above mentioned non-standard settings separately:
\begin{itemize}
    \item In the \emph{internal setting}, which is a special case of the standard setting, after $\Oh(n \log n)$ preprocessing of a length-$n$ string $T$, one can compute a representation of all covers of any substring of $T$ in $\Oh(\log n \log \log n)$ time and the shortest cover of any substring in $\Oh(\log n)$ time~\cite{DBLP:conf/spire/CrochemoreIRRSW20,DBLP:conf/cpm/BelazzouguiKPR21}.
    \item In a restricted dynamic setting in which each edit operation is reverted immediately after it is performed, the shortest cover can be updated in $\Oh(\log n)$ time~\cite{DBLP:journals/corr/abs-2402-17428}. No algorithm is known for computing covers in the fully dynamic setting.
    \item In the fully compressed setting, a representation of all covers of a length-$N$ string specified by an SLP of size $n$ with derivation tree of height $h$ can be computed in $\Oh(nh(n+\log^2 N))$ time~\cite{DBLP:journals/iandc/IMSIBTNS15}; with the technique of balancing SLPs~\cite{DBLP:journals/jacm/GanardiJL21}, the time complexity becomes $\Oh(n\log N(n+\log^2 N))$.
\end{itemize}

The $\Oh(n)$-time algorithms for computing covers of a length-$n$ string~\cite{DBLP:journals/ipl/ApostolicoFI91,DBLP:journals/ipl/Breslauer92,DBLP:journals/ipl/MooreS95} perform only single-letter comparisons and thus work also in the \pillar model. If there was a (much) more efficient algorithm computing the shortest cover of a string in the \pillar model, one would immediately improve or generalize all the above results, including our \cref{thm:main1}. We show that, contrary to the case of periods, no such efficient algorithm for covers exists. A proof of \cref{thm:main2} is given in \cref{sec:pillar}.

\begin{theorem}\label{thm:main2}
There is no algorithm in the \pillar model that solves any of the following problems for a length-$n$ binary string in $o(n/\log n)$ time:
\begin{itemize}
    \item check if $T$ is superprimitive;
    \item check if a given prefix of $T$ is a cover of $T$.
\end{itemize}
\end{theorem}

Consequently, computing the shortest cover or a representation of all covers of a string requires $\Omega(n/\log n)$ time in the \pillar model.

\section{Preliminaries}
We assume that letters of a string $T$ are numbered from 0 to $|T|-1$, i.e., $T=T[0] \cdots T[|T|-1]$. By $T[i \dd j]=T[i \dd j+1)$ we denote a substring $T[i] \cdots T[j]$. If $T$ is given in a packed form, then packed representation of its substring $T[i \dd j]$ can be computed in $\Oh((j-i+1)/\log_\sigma n)$ time using standard word RAM operations. A substring $T[i \dd j]$ is called a prefix if $i=0$ and a suffix if $j=|T|-1$. A string $B$ that occurs in $T$ as a prefix and as a suffix is called a border of $T$.

A positive integer $p$ is called a period of string $U$ if $U[i]=U[i+p]$ holds for all $i \in [0 \dd |U|-p)$. A string $U$ is called \emph{periodic} if the smallest period $p$ of $U$ satisfies $2p \le |U|$. Otherwise, $U$ is called \emph{aperiodic}. We also use the following Periodicity Lemma.

\begin{lemma}[Fine and Wilf, \cite{fine1965uniqueness}]\label{perlemma}
If a string $U$ has periods $p$ and $q$ and $p+q \le |U|$, then $\gcd(p,q)$ is a period of $U$.
\end{lemma}

For a string $X$ and non-negative integer $k$, by $X^k$ we denote a concatenation of $k$ copies of $X$. A non-empty string $U$ is \emph{primitive} if $U=X^k$ implies that $k=1$. A string of the form $X^2$ is called a square. If $X$ is primitive, the square $X^2$ is said to be \emph{primitively rooted}.

\begin{lemma}[Three Squares Lemma, \cite{DBLP:journals/algorithmica/CrochemoreR95}]\label{lem:3sq}
If a string $U$ has primitively rooted square prefixes $X^2$, $Y^2$, $Z^2$ such that $|X|<|Y|<|Z|$, then $|Z|>|X|+|Y|$.
\end{lemma}

\section{Sublinear-Time Covers}\label{sec:sublinear}

Let $c=\lfloor\frac{1}{6}\log_\sigma n\rfloor$. To show \cref{thm:main1} we divide the set $\Covers(T)$ into two subsets:
\begin{itemize}
\item $\SCovers(T)=\Covers(T)\cap [1\dd c]$ of short cover lengths, and
\item $\LCovers(T)=\Covers(T)\cap [c+1\dd n]=\Covers(T)\setminus \SCovers(T)$ of long cover lengths
\end{itemize}
and compute their representations separately. If $c=0$, there are only long covers.

\newcommand{\F}{\mathcal{F}}
\begin{lemma}\label{lem:short}
The representation of $\SCovers(T)$ can be computed in $\Oh(n/\log_\sigma n)$ time.
\end{lemma}
\begin{proof}
Let $\F$ be the set of all the factors of $T$ of length $3c$ that start at positions that are multiples of $c$. If the length of $T$ is not a multiple of $c$, when computing $\F$ we extend $T$ with at most $c$ arbitrary letters. Let us notice that every string in $\F$ fits in a machine word, so those strings can be treated as integers.

To compute $\F$, we first construct a Boolean array of all the possible length-$3c$ strings, and then iterate through the length-$3c$ substrings of $T$ starting at positions that are multiples of $c$, addressing the array directly through those integer representations. The size of the array as well as $|\F|$ is bounded by $\sigma^{3c}\le \sigma^{\frac12 \log_\sigma n}=\sqrt{n}$. Iterating through all the considered substrings takes $\Oh(n/c)$ time. Hence, this computation takes $\Oh(n/\log_\sigma n)$ time in total.

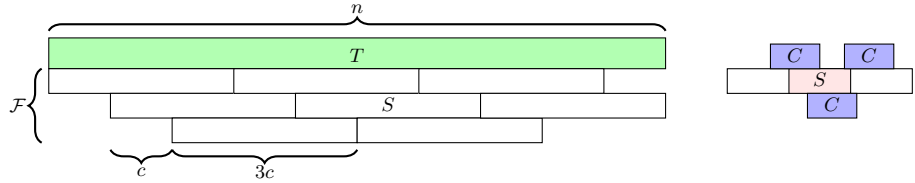
\begin{figure}[htpb]
\centering{\input{_fig_SCovers}}
\caption{Algorithm for checking a single candidate for a short cover. For each string in $\F$ we check if the occurrences of $C$ therein cover its middle part.}\label{fig:placeholder}
\end{figure}

The next claim resembles to some extent a property of seeds; cf.\ \cite[Lemma 2.2]{DBLP:journals/talg/KociumakaKRRW20}. See \cref{fig:placeholder} for an illustration.

\begin{claim}
For $i\in[1 \dd c]$, $C=T[0\dd i)$ is a cover of $T$ if and only if $C$ is a border of $T$ and the occurrences of $C$ in each string $S$ in $\F$ cover the middle length-$c$ part of $S$.
\end{claim}
\begin{proof}
$(\Rightarrow)$ Assume that $C$ is a cover of $T$. As $C$ has to cover the first and the last position of $T$, $C$ is a border of $T$. For every integer multiple $ic$ of $c$ such that $(i+1)c < |T|$, occurrences of $C$ in $T$ have to cover $U=T[ic \dd (i+1)c)$. These occurrences need to be contained in $S=T[(i-1)c \dd \min((i+2)c,|T|-1)]$. We have $S \in \F$ (possibly after appending $3c-|S|$ letters) and $C$ covers the middle length-$c$ part of $S$.

$(\Leftarrow)$ Assume that $C$ is a border of $T$ and the occurrences of $C$ in each string $S$ in $\F$ cover the middle length-$c$ part of $S$. Since $C$ is a border of $T$, its occurrences cover substrings $T[0 \dd |C|)$ and $T[|T|-|C| \dd |T|)$. Moreover, for every integer multiple $ic$ of $c$ such that $(i+1)c < |T|$, occurrences of $C$ in $T$ cover $T[ic \dd (i+1)c)$. Hence, occurrences of $C$ in $T$ cover all positions of $T$.
\qed\end{proof}

Now to compute $\SCovers(T)$ we iterate through all the $c$ lengths of candidates for a short cover independently.

For a single candidate $C=T[0\dd i)$, it is enough to check if $C$ is a suffix of $T$, which can be done in $\Oh(1)$ time, and then for each substring $S\in \F$ check if the occurrences of $C$ in $S$ cover its middle part.
The latter can be done naively in $\cO(|C|+|S|)$ time, which sums up to $\cO(c\cdot |\F|)=\Oh(\sqrt{n}\cdot \log_\sigma n)$ time for each $i$, and $\Oh(\sqrt{n}\cdot \log^2_\sigma n)=o(n/\log_\sigma n)$ time in total for all the candidates.

The result is reported in the form of the $\Oh(c)=\Oh(\log n)$ lengths of short covers (arithmetic progressions of length $1$).
\qed\end{proof}

In the computation of long covers, we use Internal Pattern Matching. In particular, we use $\IPM$ queries which, given two substrings $X$, $Y$ of $T$ such that $|X| \le |Y| \le 2|X|$, return the set of occurrences of $X$ in $Y$ represented as an arithmetic progression. Moreover, we use Period Queries that return the set of all periods (equivalently, borders) of any substring of $T$. We need to apply a period query only to $T$ itself. Such a query returns, for every $d$ being an integer power of two, the set of lengths of all borders of $T$ of length between $d$ and $2d$ represented as an arithmetic progression. If an arithmetic progression has length greater than two, then the difference $p$ of the progression is the common shortest period of all the borders represented by this progression.

\begin{theorem}[\cite{DBLP:journals/corr/KociumakaRRW13}]\label{thm:IPM}
Assume that a text $T$ of length $n$ over integer alphabet of size $\sigma$ is given in a packed form. After $\Oh(n/\log_\sigma n)$ time preprocessing, one can answer an $\IPM$ query for substrings $X$, $Y$ of $T$ in $\Oh(|Y|/|X|)$ time and a Period Query for $T$ in $\Oh(\log n)$ time.
\end{theorem}

\begin{lemma}\label{lem:long}
The representation of $\LCovers(T)$ can be computed in $\Oh(n/\log_\sigma n)$ time.
\end{lemma}
\begin{proof}
As we already noticed, a cover of $T$ is in particular its border. We ask a Period Query of \cref{thm:IPM} to compute a representation of the set of all borders of $T$. We disregard arithmetic progressions such that all their elements are smaller than $c$, as they can only correspond to the case of short covers that was already considered. Moreover, we trim the at most one remaining arithmetic progression that contains elements smaller than $c$ so that it contains only border lengths greater than $c$.

For each arithmetic progression there exists a cut-off value $t$ such that all borders of length at most $t$ represented by the progression are covers of $T$, while the longer ones are not covers of $T$. 
(This is because a shorter border from the progression is a  cover of a longer border from the progression.)
It is sufficient to compute this cut-off value for each progression.

We consider the progressions separately. Let us consider a progression $\Gamma$ of border lengths in $[d \dd 2d]$. For the two shortest borders $B_1$, $B_2$ represented by the progression $\Gamma$ (or fewer if progression $\Gamma$ contains at most one element), we use $\IPM$ queries for $B_i$ and substrings of $T$ of length $2|B_i|-1$ starting at positions $\equiv 0 \pmod{|B_i|}$ to find a representation of the set of occurrences of $B_i$ in $T$ as $\Oh(n/d)$ arithmetic progressions. This representation allows us to easily check in $\Oh(n/d)$ time if $B_i$ is a cover of $T$.

If any of $B_1$ and $B_2$ exists and is not a cover of $T$, we can safely ignore all the remaining borders in progression $\Gamma$. Otherwise, if at least three borders are represented by $\Gamma$, we know that the difference $p$ of the progression is the common smallest period of all borders represented by $\Gamma$. In $\Oh(n/d)$ time we partition the already computed occurrences of $B_1$ in $T$ into maximal arithmetic progressions of consecutive occurrences with difference $p$. Let $\Delta$ be the minimum length of such an arithmetic progression of occurrences of $B_1$. Then exactly the $\Delta$ shortest borders of the progression $\Gamma$ are covers of $T$, or all borders if the progression contains less than $\Delta$ elements.

Let us argue for the correctness of the algorithm. Let $B_1,B_2,\ldots,B_r$ be all borders represented by progression $\Gamma$ ordered by increasing lengths. We have $|B_{i+1}|=|B_i|+p$ for all $i \in [1 \dd r)$. It suffices to note that (1) $B_{\Delta'}$ for $\Delta'=\min(\Delta,r)$ is a cover of $T$ and (2) $B_{\Delta+1}$, if it exists, is not a cover of $T$.

As for (1), as all arithmetic progressions with difference $p$ of occurrences of $B_1$ in $T$ have length at least $\Delta'$, they imply occurrences of $B_{\Delta'}$ that cover the same set of positions of $T$ as the occurrences of $B_1$, i.e., all positions.

As for (2), assume that $B_{\Delta+1}$ exists (with $\Delta \ge 2$) and let $i,i+p,\ldots,i+(\Delta-1)p$ be a maximal arithmetic progression of occurrences of $B_1$ in $T$. The previous arithmetic progression has its last element smaller than $i-(|B_1|-p)$, as otherwise, by \cref{perlemma}, $B_1$ would have a period smaller than $p$. Similarly, the next arithmetic progression starts at a position greater than $i+(\Delta-1)p+(|B_1|-p)$. An occurrence of $B_{\Delta+1}$ in $T$ implies an arithmetic progression of occurrences of $B_1$ with difference $p$ and $\Delta+1$ elements starting at the same position. Thus none of the occurrences of $B_1$ at positions $i,i+p,\ldots,i+(\Delta-1)p$ extend to an occurrence of $B_{\Delta+1}$, the last position of $T$ covered by an occurrence of $B_{\Delta+1}$ from a previous arithmetic progression is smaller than $i+p$, and the first position of $T$ covered by an occurrence of $B_{\Delta+1}$ from a next arithmetic progression is greater than $i+(\Delta-1)p$. Hence, position $i+p$ of $T$ is not covered by occurrences of $B_{\Delta+1}$. This proves (2) and concludes correctness of the algorithm.

Overall, an arithmetic progression of border lengths in $[d \dd 2d]$ is processed in $\Oh(n/d)$ time. It suffices to consider $d$ such that $2d \ge c$, i.e., $d=2^i$ for $i \ge (\log_2 c) - 1$. The time complexity is thus proportional to:
$$\sum_{i=\lfloor \log_2 c \rfloor-1}^{\lfloor\log_2 n\rfloor} \frac{n}{2^i}\ \le\ \frac{n}{2^{\lfloor \log_2 c \rfloor-1}} \sum_{i=0}^{\infty} \frac{1}{2^i}\ =\ \Oh(n/c)\ =\ \Oh(n/\log_\sigma n),$$
as desired.
\qed\end{proof}

\cref{thm:main1} follows directly from \cref{lem:short,lem:long}.

\begin{remark}
The algorithm for computing long covers works in $\Oh(n/\log_\sigma n)$ time in the \pillar model. For a constant $\sigma$, the complexity matches our lower bound of \cref{thm:main2}. However, the computation of short covers in \cref{lem:short} works in $\Theta(n)$ time in the \pillar model.
\end{remark}

\section{Sublinear Data Structure for Cover Array}\label{sec:prefixes}

\subsection{Why Representing the Cover Array in Sublinear Space can be a Challenge}
The cover array may require $\Theta(n \log n)$ bits to represent in a straightforward manner. In particular, the array may contain $\Theta(n)$ different values; this is true even if we disregard trivial positions $i$ such that $\Cov_T[i]=i$ and positions $i$ such that $T[0 \dd i)$ is periodic, as shown in the following \cref{ex:ababa}.

\begin{example}\label{ex:ababa}
Let $T=\mathtt{a}^{2m}\mathtt{b}\mathtt{a}^{3m}\mathtt{b}\mathtt{a}^{2m}$ for positive integer $m$ be a string of length $\Theta(m)$. Then all prefixes of $T$ of length at least $2m+1$ are aperiodic and the last $m+1$ positions of the array $\Cov_T$ contain the following lengths of proper covers: $3m+1,3m+2,\ldots,4m+1$.
\end{example}

\cref{ex:ababa} might still not be fully convincing that a sublinear-sized representation of the cover array is not obvious. Indeed, the the cover array of the string family from \cref{ex:ababa} has an especially simple structure (a prefix consisting only of ones, a substring with an arithmetic sequence with difference 1 corresponding to superprimitive prefixes, and a suffix with arithmetic sequence with difference 1). Below in \cref{cor} we give a different example, that in a  Fibonacci string all but a logarithmic number of prefixes have a proper cover and (except for a short prefix) no \emph{two} consecutive positions of the cover array form an arithmetic sequence of difference 1. The cover array of a Fibonacci string contains a logarithmic number of different values.

Let us recall that the Fibonacci strings are defined as follows: $\Fib_0=\mathtt{b}$, $\Fib_1=\mathtt{a}$, and $\Fib_m=\Fib_{m-1}\Fib_{m-2}$ for $m>1$. All covers of whole Fibonacci strings (as well as other types of quasiperiodicity) were characterized in~\cite{DBLP:journals/arscom/ChristouCI16} (see also \cite{singh:hal-02141636} for similar results on Tribonacci strings). Moreover, a complete characterization of the lengths of shortest covers of cyclic shifts of Fibonacci strings was shown~\cite{DBLP:journals/tcs/CrochemoreIRRSW21}. However, apparently, the structure of the cover array of Fibonacci strings was not studied before. The theorem below shows the recursive structure of the array; see \cref{fig:cov_fib} for a concrete example.

Let $\Fib$ be the infinite Fibonacci string (the limit of strings $\Fib_m$). For any $m>0$, $\Fib_m$ is a prefix of $\Fib$, and hence also $\Cov_{\Fib_m}$ is a prefix of $\Cov_{\Fib}$. Thus it is enough to characterize the values of $\Cov_{\Fib}$. Let $F_k=|\Fib_k|$.

\begin{theorem}\label{thm:fib}
In the corner cases $\Cov_{\Fib}[\ell]$ is equal to
\begin{itemize}
\item $\ell$ if $\ell\le 2$,
\item $3$ if $\ell=F_k$ for odd $k\ge 3$,
\item $5$ if $\ell=F_k$ for even $k\ge 4$,
\item $\ell$ if $\ell=F_{k}-1$ or $\ell=2F_{k}-1$ for $k \ge 4$.
\end{itemize}
Otherwise, $\Cov_{\Fib}[\ell]=\Cov_{\Fib}[\ell-F_{k-1}]$, where  $F_k<\ell<F_{k+1}$.
\end{theorem}
\begin{proof}
It is well-known that for any $m \ge 1$, $\LCP(\Fib_{m+1},\Fib_{m-1}\Fib_m)=F_{m+1}-2$; see e.g.\ \cite{DBLP:journals/siamcomp/KnuthMP77}.
In particular, $F_{m-1}$ is a period of $\Fib_{m+1}[0 \dd F_{m+1}-2) = (\Fib_{m-1}\Fib_m)[0 \dd F_{m+1}-2) = (\Fib_{m-1}\Fib_{m-1}\Fib_{m-2})[0 \dd F_{m+1}-2)$, but not a period of $\Fib_{m+1}[0 \dd F_{m+1}-1)$.
Equivalently, $\Fib[0\dd \ell-F_{m-1})$ is a border of $\Fib[0\dd \ell)$ if and only if $F_{m-1}<\ell\le F_{m+1}-2$.

The value of $\Cov_{\Fib}[F_k]$ as well as of $\Cov_{\Fib}[\ell]$ for $\ell\le 2$ follows from \cite{DBLP:journals/arscom/ChristouCI16}. From the same paper we know that $\Fib_{k-2}$ is the longest proper border of $\Fib_k$. Moreover, $\Fib[0\dd F_{k-2}-1)$ is the longest border of $\Fib[0\dd F_k-1)$. Indeed,
an existence of a longer border (of length different than $F_{k-1}-1$) would result in $\Fib[0\dd F_k-2)$ having period $1$ by the periodicity lemma (as it already has periods $F_{k-1}$ and $F_{k-2}$; see \cite{DBLP:journals/siamcomp/KnuthMP77}).

We will prove by induction that $\Cov_{\Fib}[F_k-1]=F_k-1$ for $k \ge 4$. The base case holds. By the above, the only candidate for the length of a proper cover of $\Cov_{\Fib}[F_k-1]$ is $\Cov_{\Fib}[F_{k-2}-1]$, which equals $F_{k-2}-1$ by induction.
Prefix $\Fib[0\dd F_{k-2}-1)$ has an occurrence at positions $0$ and $F_{k-2}$ in this $\Fib$,
but the position $F_{k-2}-1$ remains uncovered; existence of yet another occurrence that contains this position would result in a long overlap of occurrences which, in turn, would result in the string $\Fib[0\dd F_{k-2}-1)$ being periodic, which is not the case. Hence, $\Cov_{\Fib}[F_k-1]=F_k-1$.

Next we prove by induction that $\Cov_{\Fib}[2F_k-1]=2F_k-1$ for $k \ge 4$. Similarly, $\Fib[0\dd F_k-1)$ is the longest border of $\Fib[0\dd 2F_k-1)$. By the inductive hypothesis, $\Cov_{\Fib}[F_k-1]=F_k-1$ is the only candidate for the length of a proper cover of $\Fib[0\dd 2F_k-1)$.
By exactly the same argument as in the previous case, the position $F_k-1$ in $\Fib$ is not covered by this candidate. Thus $\Cov_{\Fib}[2F_k-1]=2F_k-1$.

We have $\Cov_{\Fib}[6]=3$. Now, for $k \ge 5$, let $\ell\in[F_k+1\dd 2F_{k-1}-2]\cup [2F_{k-1}\dd F_{k+1}-2]$.
As noted, $\Fib[0\dd \ell-F_{k-1})$ is a border of $\Fib[0\dd \ell)$ (since $\ell\le F_{k+1}-2$). Additionally, if $\ell\le 2F_{k-1}-2$, string $\Fib[0\dd \ell-F_{k-1})$ also appears in $\Fib$ at position $F_{k-2}$ (by the $\LCP$ equality from the beginning of the proof). Those two or three occurrences cover all the positions of $\Fib[0\dd \ell)$, hence a cover of $\Fib[0\dd \ell-F_{k-1})$ is also a cover of $\Fib[0\dd \ell)$. At the same time a shortest cover of $\Fib[0\dd \ell)$ has to be a cover of a border that is a cover, hence $\Cov_{\Fib}[\ell]=\Cov_{\Fib}[\ell-F_{k-1}]$.
\qed\end{proof}

\begin{corollary}\label{cor}
For any $m \ge 1$, the array $\Cov_{\Fib_m}$ contains $\Theta(m)$ different values. Only $\Theta(m)$ prefixes of $\Fib_m$ are superprimitive. Moreover, for all $\ell \in [5 \dd F_m)$, we have $\Cov_{\Fib_m}[\ell]+1 \ne \Cov_{\Fib_m}[\ell+1]$.
\end{corollary}
\begin{proof}
The first two statements follow readily from \cref{thm:fib}. As for the third statement, among the distinct values in $\Cov_{\Fib}$ from position 5 onwards, the only pairs of consecutive numbers are $(3,4)$ and $(4,5)$. (This is because for large enough $k$, values $F_{k+1}-1$ and $2F_k-1$ differ by more than 1.) Therefore, if $\ell \ge 5$ would be the smallest position such that $\Cov_{\Fib}[\ell+1]=\Cov_{\Fib}[\ell]+1$, then $\ell$ or $\ell+1$ would be equal to $F_k$ for some $k \ge 5$. By the recursion in \cref{thm:fib}, if $\ell+1=F_k$, then $\Cov_\Fib[\ell]=\ell>5$ and $\Cov_\Fib[\ell+1] \in \{3,5\}$, so two consecutive values are not possible. If $\ell=F_k$ and $\Cov_\Fib[\ell]=3$, then $\Cov_\Fib[\ell+1]=9$ by easy induction, so again two consecutive values on consecutive positions are not possible.
\qed\end{proof}

The recursive characterization of \cref{thm:fib} allows to compute any element of the cover array of $\Fib_m$ in $\Oh(\log n)$ time, where $n=F_m$, without additional space. 
By \cref{cor}, the cover array of $\Fib_m$ has only $\Oh(\log n)$ different values, which allows one to store the cover array of $\Fib_m$ in a packed form in $\Oh(n \log \log n / \log n)$ space so that its elements can be retrieved in $\Oh(1)$ time. In the next subsection we show that an equally space-efficient representation exists for every string over a constant-sized alphabet.

\newcommand{\pref}{\mathit{pref}}
\renewcommand{\sp}{\mathit{sp}}
\subsection{Proof of \cref{thm:online}}
We use the following known corollary of the periodicity lemma.

\begin{lemma}[\cite{DBLP:journals/algorithmica/BreslauerG95,DBLP:conf/icalp/PlandowskiR98}]\label{lem:3occs}
If $|X| < |Y| < 2|X|$ are two strings and $X$ has at least three occurrences in $Y$ as a substring, then $X$ is periodic.
\end{lemma}

Before we describe the data structure, let us give some intuition.

Assume that $\Cov_T[\ell]=c$ with $c < \ell$. That is, string $C=T[0 \dd c)$ is a proper shortest cover of a prefix $T[0 \dd \ell)$. If the second occurrence of $C$ in $T$ is at position $j>0$, then $U^2=T[0 \dd 2j)$ is a square. Further, $j>c/2$, as otherwise $C$ would be periodic. Hence, $C$ is a prefix of $T[0 \dd 2j)$. This concludes that the square $T[0 \dd 2j)$ is primitively rooted, as otherwise $C$ would be periodic. By \cref{lem:3sq}, there are only $\Oh(\log n)$ primitively rooted square prefixes of $T$. Thus, if $C$ is a proper shortest cover of a prefix of $T$, we can assign to $C$ one of $\Oh(\log n)$ primitively rooted square prefixes of $T$.

Let $P=T[0 \dd p)$ be the shortest aperiodic prefix of $T$ such that $p \ge j$. As $T[0 \dd c)=C$ is aperiodic, $p$ is well-defined and $C$ has a prefix $P$. Thus, if $C$ is a proper shortest cover of a prefix of $T$, this allows to assign to $C$ one of $\Oh(\log n)$ aperiodic prefixes of $T$.

For $k=\ell-c$, we have $T[k \dd k+p)=P$. We observe that there can be no further occurrence of $P$ in $T$ at a position in $(k\dd \ell-p]$ (that is, no further occurrence of $P$ in $T[0 \dd c)$). Indeed, such an occurrence would be a substring of $C$, so it would imply an occurrence of $P$ in $T$ at a position in $[1 \dd j)$. By \cref{lem:3occs}, this would contradict the fact that $P$ is aperiodic.
In summary, if $C$ is a proper shortest cover of a prefix $T[0 \dd \ell)$, then $C$ can be uniquely identified by the rightmost occurrence in $T[0 \dd \ell)$ of the aperiodic prefix $P$ of $T$ that is assigned to $C$. Moreover, the occurrence is at one of the positions in $(\ell-2p\dd \ell)$, as $2p \ge 2j > c$.

\textbf{Data structure:} Let $j_1,\ldots,j_t$ be the half lengths of all primitively rooted square prefixes of $T$. By \cref{lem:3sq}, we have $t = \Oh(\log n)$. The data structure stores $t$ lengths of aperiodic prefixes of $T$, $p_1,\ldots,p_t$. For every $i \in [1 \dd t]$, $p_i$ is the length of the shortest aperiodic prefix of $T[0 \dd 2j_i)$ of length at least $j_i$. (It is known that such a prefix exists, as $T[0 \dd 2j_i-1)$ is aperiodic by \cref{perlemma}.)

For each $\ell \in [1 \dd n]$, we store a bit $\sp[\ell]$ that equals 1 if and only if $T[0 \dd \ell)$ is superprimitive. If $\sp[\ell]=0$, a number $\pref[\ell] \in [1 \dd t]$ is stored that determines the aperiodic prefix $P=T[0 \dd p_{\pref[\ell]})$ of $T$ that corresponds to the shortest cover $C$ of $T[0 \dd \ell)$, as discussed above. Precisely, if $j_i$ is the position of the second occurrence of $C$ in $T$, then $\pref[\ell]=i$. Finally, a data structure for $\IPM$ queries in $T$ is stored.

Overall, provided that the arrays $\sp$ and $\pref$ are stored in a packed form, the space complexity is $\Oh(\log n + n \log \log n/\log n+n/\log n + n/\log_\sigma n) = \Oh(n \log \log n/\log n+n/\log_\sigma n)$, as required.

\textbf{Queries:} To compute $\Cov_T[\ell]$, we first check if $\sp[\ell]=1$ and, if that is the case, return $\ell$. Otherwise, we ask an $\IPM$ query to compute the righmost occurrence of $P=T[0 \dd p_{\pref[\ell]})$ in $T[\ell-2\cdot p_{\pref[\ell]}+1 \dd \ell)$. As $P$ is aperiodic, there are at most two such occurrences. We select as $k$ the starting position of the rightmost occurrence. The shortest cover of $T[0 \dd \ell)$ is $T[k \dd \ell)$ (i.e., $\Cov_T[\ell]=\ell-k$).

By \cite{DBLP:journals/corr/KociumakaRRW13}, the query time complexity is $\Oh(1)$. This concludes the proof of \cref{thm:online}.

\section{Lower Bound on the Complexity in the \pillar Model}\label{sec:pillar}

\subsection{The \pillar model}
Let us start first formally introduce the primitives of the \pillar model~\cite{FOCS20}. The argument strings are fragments of strings in a given collection $\mathcal{X}$:
\begin{itemize}
\item $\mathsf{Extract}(S, \ell, r)$: Retrieve string $S[\ell\dd r)$.
\item $\LCP(X, Y),\, \LCPR(X, Y)$: Compute the length of the longest common prefix/suffix of $X$ and $Y$.
\item $\IPM(X, Y)$: Assuming that $|Y| \le 2|X|$, compute the starting positions of all exact occurrences of $X$ in $Y$, expressed as an arithmetic progression.
\item $\mathsf{Access}(S, i)$: Retrieve the letter $S[i]$;
\item $\mathsf{Length}(S)$: Compute the length $|S|$ of the string $S$.
\end{itemize}
The runtime of algorithms in this model can be expressed in terms of the number of primitive \pillar operations (and additional operations not performed on the strings themselves).

\subsection{Lower Bound}
We focus on checking if a string over an alphabet $\{\mathtt{a},\mathtt{b}\}$ is covered by its border $\mathtt{aba}$. Strings covered by $\mathtt{aba}$ are formed of concatenations of strings of a form $(\mathtt{ab})^k\mathtt{a}$ for $k\ge 1$; equivalently, strings that have $\mathtt{aba}$ as a border and do not contain a substring $\mathtt{bb}$ or $\mathtt{aaa}$. 

For infinitely many positive integers $n$, we show a strategy for an adversary to answer $Cn/\log n$ \pillar queries on a length-$n$ binary string, for a certain constant $C>0$, after which the adversary still has the choice of fixing the string in two ways: in one $T$ has a cover $\mathtt{aba}$, and in the other $T$ is superprimitive (i.e., it has no proper cover).

We define a morphism $\phi:\{\mathtt{0},\mathtt{1}\} \mapsto \{\mathtt{a},\mathtt{b}\}$:
\begin{itemize}
\item $\phi(\mathtt{0})= \mathtt{abababa}$ $\mathtt{aba}$ $\mathtt{ababa}=(\mathtt{ab})^3\mathtt{a}(\mathtt{ab})\mathtt{a}(\mathtt{ab})^2\mathtt{a}$
\item $\phi(\mathtt{1})=\mathtt{abababa}$ $\mathtt{ababa}$ $\mathtt{aba}=(\mathtt{ab})^3\mathtt{a}(\mathtt{ab})^2\mathtt{a}(\mathtt{ab})\mathtt{a}$
\end{itemize}
Both $\phi(\mathtt{0})$ and $\phi(\mathtt{1})$ have length $15$ and have a cover $\mathtt{aba}$. Thus $\phi(S)$, for any string $S$ over alphabet $\{\mathtt{0},\mathtt{1}\}$, has a cover $\mathtt{aba}$.

Let us recall that a de Bruijn sequence of order $k$ over an alphabet $\Sigma$ is a string of length $|\Sigma|^k+k-1$ over the alphabet $\Sigma$ such that its every substring of length $k$ is distinct. It is well-known that such sequences exist for every finite alphabet $\Sigma$ and integer $k \ge 1$~\cite{deBruijn}.

Let $B_k$ be a de Bruijn sequence of order $k$ over the binary alphabet $\{\mathtt{0},\mathtt{1}\}$. We apply the morphism $\phi$ on $B$ to obtain a string $T_k$ over alphabet $\{\mathtt{a},\mathtt{b}\}$ of length $15\cdot(2^k+k-1)$. Due to the property of de Bruijn sequences, each substring of $T$ of length at least $15(k+1)-1$ is distinct. Indeed, every ``aligned'' substring of length $15k$ starting at a position divisible by 15 in $T_k$ is distinct, and every substring of $T_k$ of length $15(k+1)-1$ contains an ``aligned'' substring of length $15k$.

Due to this property, an answer to an $\LCP$ or $\LCPR$ query on $T_k$ for two different positions is always bounded from above by $15(k+1)-2$. Similarly for the $\IPM$ queries; if we query for a substring of length at least $15(k+1)-1$, then we do not gain any interesting information (the only occurrence of the substring is the one used to ask the query). On the other hand, by asking an $\IPM$ query for a shorter substring we only gain information about a part of $T_k$ of length at most $30(k+1)-2$.

Formally, the strategy of the adversary for a text $T$ of length $n=|T_k|$ is as follows. Queries $\mathsf{Extract}$, $\mathsf{Access}$, $\mathsf{Length}$, $\LCP$, $\LCPR$ are answered as in $T_k$. An $\IPM(X,Y)$ query for $|X| < 15(k+1)-1$ is also answered as in $T_k$. Finally, to answer an $\IPM(X,Y)$ query for $|X| \ge 15(k+1)-1$, we refer to the fragments $T[i_x \dd j_x]=X$ and $T[i_y \dd j_y]=Y$ and return an occurrence of $X$ in $Y$ at position $i_y-i_x$ if $[i_x \dd j_x] \subseteq [i_y \dd j_y]$ and no occurrence otherwise.

We say that a position $i \in [0 \dd n)$ of $T$ \emph{has been touched} if the algorithm has performed (1) an $\mathsf{Access}$ query on $T[i]$, or (2) an $\LCP(T[i_x \dd j_x],T[i_y \dd j_y])$ query such that $i \in [i_x \dd i_x+\ell) \cup [i_y \dd i_y+\ell)$ where $\ell$ is the result of the $\LCP$ query, or (3) similarly an $\LCPR$ query such that $i$ belongs to the computed $\LCPR$ of one of the two queried substrings of $T$, or (4) an $\IPM(T[i_x \dd j_x],T[i_y \dd j_y])$ query for $j_x-i_x+1 < 15(k+1)-1$ such that $i \in [i_x \dd j_x] \cup [i_y \dd j_y]$. In total, after $q$ \pillar operations, fewer than $45q(k+1) \le 90kq$ positions of $T$ have been touched. Thus after $q = \lfloor 2^k/(6k)\rfloor$ operations, there still exists a position in $T$ that has not been touched. Assume $i$ is such a position. Then the adversary can make the choice to set $T[i]$ as $T_k[i]$ or as the letter different from $T_k[i]$; all the remaining untouched positions are set as in $T_k$. If $T[i]=T_k[i]$, $T=T_k$ has a cover $\mathtt{aba}$. If $T[i] \ne T_k[i]$, $T$ contains exactly one substring $\mathtt{a}^s$ for some $s \in [3 \dd 5]$, or exactly one substring $\mathtt{b}^t$, for some $t \in [2 \dd 3]$. It is easy to see that in this case $T$ is superprimitive.

There exists a constant $C>0$ (for example, $C=1/180$) such that the selected value of $q$ satisfies $q \ge Cn/\log n$. \cref{thm:main2} is proved.

\section{Open Problems}\label{sec:op}
It remains open if the data structure of \cref{thm:online} can be constructed in sublinear time or if its space complexity can be decreased to $\Oh(n/\log n)$ for $\sigma = \Oh(1)$.

Future work also includes designing sublinear-time algorithms for other notions of quasiperiodicity for which $\Oh(n)$-time algorithms are already known, for a length-$n$ string over an integer alphabet; this includes, for example, seeds~\cite{DBLP:journals/talg/KociumakaKRRW20,DBLP:conf/esa/Radoszewski23}, enhanced covers~\cite{DBLP:journals/tcs/FlouriIKPPST13}, and partial covers~\cite{DBLP:conf/esa/Radoszewski23}.

\bibliographystyle{splncs04}
\bibliography{references}

\end{document}

%% file: _fig_SCovers.tex
\begin{tikzpicture}[scale=0.82,transform shape]

\tikzstyle{dot}=[inner sep=0.045cm, circle, draw, fill=red]

\draw[fill=green!30!white] (0,0) rectangle +(10, 0.5);
\node [above] at (5, 0) {$T$};
\draw[thick,decorate,decoration={brace,raise=5pt,amplitude=5pt}] 
      (0,0.4)--(10,0.4) node[midway,above,yshift=11pt] {$n$};

\draw[fill=white] (0,0) rectangle +(3, -0.4); 
\draw[fill=white] (1,-0.4) rectangle +(3, -0.4); 
\draw[fill=white] (2,-0.8) rectangle +(3, -0.4);  
\draw[fill=white] (3,0) rectangle +(3, -0.4);  
\draw[fill=white] (4,-0.4) rectangle +(3, -0.4); 
\draw[fill=white] (5,-0.8) rectangle +(3, -0.4);  
\draw[fill=white] (6,0) rectangle +(3, -0.4);
\draw[fill=white] (7,-0.4) rectangle +(3, -0.4);
\node [above] at (5.5, -0.8) {$S$};
\draw[thick,decorate,decoration={brace,raise=5pt,amplitude=5pt}] 
      (0.1,-1.2)--(0.1,0) node[midway,left,xshift=-10pt] {$\mathcal{F}$};
\draw[thick,decorate,decoration={brace,raise=5pt,amplitude=5pt}] 
      (2,-1.1)--(1,-1.1) node[midway,below,yshift=-11pt] {$c$};
\draw[thick,decorate,decoration={brace,raise=5pt,amplitude=5pt}] 
      (5,-1.1)--(2,-1.1) node[midway,below,yshift=-10pt] {$3c$};

\draw[fill=white] (11,0) rectangle +(3, -0.4);
\draw[fill=red!10!white] (12,0) rectangle +(1, -0.4);
\node [above] at (12.5, -0.4) {$S$};

\draw[fill=blue!30!white] (11.7,0) rectangle +(0.8, 0.4); 
\node [above] at (12.1, 0) {$C$};
\draw[fill=blue!30!white] (12.3,-0.4) rectangle +(0.8, -0.4); 
\node [above] at (12.7, -0.8) {$C$};

\draw[fill=blue!30!white] (12.9,0) rectangle +(0.8, 0.4); 
\node [above] at (13.3, 0) {$C$};
\end{tikzpicture}